\documentclass[conference]{IEEEtran}

\usepackage[dvips]{color}
\usepackage{epsf}
\usepackage{times}
\usepackage{epsfig}
\usepackage{graphicx}
\usepackage{amsmath}
\usepackage{amssymb}
\usepackage{amsxtra}
\usepackage{here}
\usepackage{rawfonts}
\usepackage{times}
\usepackage{url}
\usepackage{cite}
\usepackage{tabu}
\usepackage{graphicx}
\usepackage{geometry,color,graphicx}
\usepackage{epstopdf}
\usepackage{multirow}
\usepackage{fancyhdr}
\usepackage[utf8]{inputenc}
\usepackage[english]{babel}

\usepackage{amsthm}

\DeclareGraphicsExtensions{.eps}

\newtheorem{theorem}{\bf Theorem}

\newtheorem{proposition}{\bf Proposition}

\newtheorem{definition}{\bf Definition}

\newlength{\aligntop}
\newlength{\alignbot}
\makeatletter
\renewenvironment{align}{%
  \vspace{\aligntop}
  \start@align\@ne\st@rredfalse\m@ne
}{%
  \math@cr \black@\totwidth@
  \egroup
  \ifingather@
    \restorealignstate@
    \egroup
    \nonumber
    \ifnum0=`{\fi\iffalse}\fi
  \else
    $$%
  \fi
  \ignorespacesafterend%
  \vspace{\alignbot}\par\noindent
}

\IEEEoverridecommandlockouts

\makeatletter
\def\footnoterule{\kern-3\p@
  \hrule  \kern 4\p@} 
\makeatother

\begin{document}
\title{\huge Prospect Theory for Enhanced Smart Grid Resilience Using Distributed Energy Storage
 \vspace{-0.55cm}}

\author{\IEEEauthorblockN{Georges El Rahi$^1$, Anibal Sanjab$^1$, Walid Saad$^1$, Narayan B. Mandayam$^2$, and H. Vincent Poor$^3$  \thanks{This research was supported in part by the National Science Foundation under Grants ECCS-1549894, ECCS-1549900, CNS-1446621, ACI-1541105, ACI-1541069, and  ECCS-1549881.} }\IEEEauthorblockA{\small 
$^1$ Electrical and Computer Engineering Department, Virginia Tech, Blacksburg, VA, USA, Emails:  \{gelrahi, anibals, walids\}@vt.edu \\
$^2$ Electrical and Computer Engineering Department, Rutgers University, North Brunswick, NJ, USA, Email: \url{narayan@winlab.rutger.edu}\\
$^3$ Electrical Engineering Department, Princeton University, Princeton, NJ, USA, Email: \url{poor@princeton.edu} 
\vspace{-0.60cm}
 }%
 }
\maketitle 

\begin{abstract}
The proliferation of distributed generation and storage units is leading to the development of local, small-scale distribution grids, known as microgrids (MGs). In this paper, the problem of optimizing the energy trading decisions of MG operators (MGOs) is studied using game theory. In the formulated game, each MGO chooses the amount of energy that must be sold immediately or stored for future emergencies, given the prospective market prices which are influenced by other MGOs' decisions. The problem is modeled using a Bayesian game to account for the incomplete information that MGOs have about each others' levels of surplus. 
The proposed game explicitly accounts for each MGO's subjective decision when faced with the uncertainty of its opponents' energy surplus. In particular, the so-called framing effect, from the framework of prospect theory (PT), is used to account for each MGO's valuation of its gains and losses with respect to an individual utility reference point. The reference point is typically different for each individual and originates from its past experiences and future aspirations. A closed-form  expression for the Bayesian Nash equilibrium is derived for the standard game formulation. Under PT, a best response algorithm is proposed to find the equilibrium. Simulation results show that, depending on their individual reference points, MGOs can tend to store more or less energy under PT compared to classical game theory. In addition, the impact of the reference point is found to be more prominent as the emergency price set by the power company increases.

\end{abstract}


\vspace{-0.2cm}
\section {Introduction}
\vspace{-0.10cm}

The emerging concept of microgrids (MGs) will play a major role in the modernization of the power grid. Microgrids are small-scale local power grids which are, typically, composed of renewable generation units, storage devices, and energy consumers \cite{path}. MGs are managed by various MG operators (MGOs) and can operate in either connected or islanded modes, and are expected to bring forth innovative solutions for the smart grid by enhancing power management and providing energy reserves via storage. 

\indent Indeed, the storage capability of MGs can be used to assist in the energy management of the smart grid as investigated by a number of recent works \cite{scutari, scutari2, hamed}. 
%
%
However, more recently, there has been considerable interest in using the storage abilities of MGs  to enhance the resilience of the smart grid against emergency events such as natural disasters or security breaches.  In this regard, various academic, industrial, and federal reports \cite{intro1,intro2,intro3} have proposed leveraging the MGs' storage capacity to mitigate the effect of loss of generation during emergencies by meeting the smart grid's most critical loads. Indeed, distributed storage and generation units, the integral constituents of MGs, have played an essential role in preserving the operation of hospitals and police stations, as well as fire fighting and rescue services centers in many recent emergency situations in the United States \cite{intro3}. For instance, this has been the case during natural disasters such as hurricanes Katrina and Rita, and the wildfires which interrupted the transmission of electricity to parts of Utah in 1995 and 2003, as well as in the 2003 North American Northeast blackout \cite{intro3}. In addition to the various reports in \cite{intro1,intro2,intro3} that encourage the use of MG storage to enhance grid resilience, other works such as \cite{power1} and \cite{power2} have also investigated the issues related to power quality that might arise when a critical load is supplied by MG energy sources. However,  there is a lack of works that analyze the willingness and ability of MGOs to participate in covering the power grid's critical loads.   \\
\indent To this end, in order to leverage the distributed storage units across MGs, the power companies must offer significant financial incentives for the MGOs to keep a portion of their energy surplus in storage for potential emergency use. The MGOs are hence faced with the choice of selling their excess at the current market price, or storing it and potentially selling it at the significantly higher emergency price, in the future. Moreover, given the fact that the energy bought in case of emergency is limited, competition will arise between the different MGOs who seek  to take advantage of the incentives offered by the power company for emergency energy.    \\
\indent In this regard, game theory \cite{walid} can be used to model the interdependency
between MGOs and predict the outcomes of their competitive behavior. In fact, game-theoretic analysis has been a popular tool for understanding the interactions between  storage owners in smart grid energy management \cite{scutari, scutari2, hamed}. However, these works do not investigate the aforementioned scenarios in which storage is used for improving resilience. Moreover, these works typically rely on games with complete information, which are not practical for smart grid scenarios.\\
\indent Another key drawback of existing game-theoretic analysis is the assumption that all players are rational and  thus seek to maximize their expected utilities in a similar objective manner. In a real-life application however, as observed by the experimental studies in \cite{nobel} and \cite{CPT}, the behavior of individuals can deviate considerably from the rational principles of conventional game theory.  In this regard, the framework of \emph{prospect theory} (PT) \cite{nobel} can be used to model the non-rational behavior of MGOs in the presence of uncertainty such as renewable energy sources \cite{proc},  and its impact on the ability of MGs to meet the power grid's critical load.\\
\indent The main contribution of this paper is to propose a new framework for analyzing the storage strategy of MGOs in order to enhance smart grid resilience. In this regard, we formulate a noncooperative Bayesian game between multiple MGOs  to account for the incomplete information of each MGO regarding the excess of energy of its opponents. In this game, each MGO must choose a portion of its MG's energy excess to store so as to maximize a utility function that captures the tradeoff between selling at the current market price and potentially selling in the future at a significantly higher emergency energy price. In contrast to conventional game theory, we develop a prospect-theoretic framework that models the behavior of MGOs when faced with the uncertainty of their opponents' stored energy, which stems from the presence of intermittent renewable energy sources. In particular, we account for each MGO's valuation of its gains and losses with respect to its own individual utility evaluation perspective, as captured via the PT framing effect \cite{nobel} by a utility reference point. This reference point represents a utility that an individual MGO anticipates and it originates from previous experiences and future aspirations
of profits, which can differ in between MGOs \cite{CPT}. \\
\indent For this proposed game, we derive the  closed-form expression for the Bayesian Nash equilibrium (BNE) for the classical game-theoretic scenario and interpret this equilibrium under different conditions. For the PT case, we propose a best response algorithm  that allows the MGOs to reach a BNE in a decentralized fashion. 
Simulation results highlight the difference in MGO behavior between the fully rational case of classical game theory (CGT) and the prospect-theoretic scenario. Indeed, for certain reference points, MGOs choose to store more energy under PT compared to CGT, while the case is reversed for other reference points where MGOs noticeably reduce their MGs' stored energy. In addition, the impact of the reference point is found to be more prominent as the emergency price increases. The power company  must therefore quantify the subjective behavior of the MGOs before choosing the optimal emergency energy price, in order to meet the critical load at minimal cost. \\
\indent {The rest of this paper is organized as follows. Section II presents  the system model and provides the Bayesian game formulation.  In Section III, we present the game solution under classical game theoretic analysis, while we present in Section IV  the game solution under prospect theoretic analysis. In Section V, we present and interpret our simulation results, and finally conclusions are drawn in Section VI.

\vspace{-0.08cm}
\section{System model and Bayesian Game Formulation}

Consider a large-scale smart grid managed by a power company that integrates a set $\mathcal{N}$ of $N$  microgrids, each of which is managed by an MG operator. Microgrids are small-scale distribution grids which typically include renewable generation units, storage devices, and energy consumers. Each MG operator manages all energy trades conducted by its own MG. 
Each MG $n \in \mathcal{N}$, managed by its MGO $n$, includes a storage unit with capacity $Q_{n,\textrm{max}}$ which can be used to store the excess of energy produced. Given the intermittent nature of renewable energy sources, each MG's energy surplus $Q_n \in \left[ 0, Q_{n,\textrm{max}} \right] $ is unknown beforehand and will vary over time. 
A positive $Q_n$ indicates that an MG has extra energy while $Q_n= 0$  indicates that no surplus is available. Given an amount of energy surplus, $Q_n$, an MGO $n$ has the option of selling this stored energy to the grid at the corresponding retail price, $\rho$, or saving it for later use in case of emergency, for improved resilience. In this regard, each MGO will choose a portion $\alpha_n \in \left[0 , 1 \right]$ of its MG's $Q_n$ to store and will consequently sell the rest. In case of emergency or blackout, the power company will
purchase the stored energy  to cover a certain required
critical load $L_c$, until normal power supply is restored.

In order to increase the resilience of the power grid  against emergency events,  the power company will encourage the MGOs to store part 
of their MGs' excess by offering a price $\rho_{c}$  per unit of stored energy purchased in case of emergency.  Typically, $\rho_{c}$ must be significantly larger than $\rho$ to incentivize the MGOs to store the excess. If the total stored energy exceeds the needed $L_c$, the power company will no longer purchase the entire energy stored by each MG.

Let $\boldsymbol{\alpha}$ and $\boldsymbol{Q}$  be the vectors that represent, respectively, the storage strategy and the available energy surpluses of all the MGOs in the set  $\mathcal{N}$. 
In this respect, when $\boldsymbol{\alpha}^\intercal\boldsymbol{Q} > L_c$, the power company will purchase, from each MG $n$, an amount of energy $D_n$ given by:

\vspace{-0.2cm}
\begin{align}
D_n =\left(\alpha_n Q_n - \frac { \boldsymbol{\alpha}^\intercal\boldsymbol{Q} -L_c }  {\textit{N}} \right)^{{+}}, 
\end{align}
where
$(q)^+ = \textrm{max}(0, q)$. $ \boldsymbol{\alpha}^\intercal\boldsymbol{Q} -L_c  $ is the amount by which the total stored energy exceeds the required $L_c$. Let $\theta$ be the expected probability of an emergency event occurring. Then, each MGO $n$ will choose its optimal storage strategy $\alpha_n$ to optimize the following utility function:
 \begin{align} \label{eq:uti}
\small{
    U_n(\boldsymbol{\alpha},\boldsymbol{Q}) =
    \begin{cases}
     \rho\left( Q_n - \alpha_n Q_n \right) + \theta \rho_{c} \alpha_n Q_n,  & \text{if}\   \boldsymbol{\alpha}^\intercal\boldsymbol{Q} \leq L_c , \\
     \rho\left( Q_n - \alpha_n Q_n\right) + \theta \rho_{c} D_n,  & \text{otherwise.}
    \end{cases}
}
 \end{align}
\normalsize
\indent  Note that, when $\theta \rho_{c} < \rho$, the MGOs will have no incentive to store their MGs' excess and, hence, they will sell all the available surplus at the current market price. Thus, hereinafter, we restrict our analysis to the case $\theta \rho_{c} > \rho$. 
As seen from (\ref{eq:uti}), the driving factor in determining an MGO's optimal strategy is the total energy stored by its opponents. In fact, as $\boldsymbol{\alpha}^\intercal\boldsymbol{Q} -L_c$ increases, so will the amount of stored energy which will not be bought in case of emergency. Indeed, the MGO could have instead sold that energy at the current market price and made a profit. Given this trade-off between selling at the current market price and storing the excess for a potentially higher profit in case of emergency, each MGO aims at maximizing its utility function by choosing the optimal storage strategy $\alpha_n$, while also accounting for the actions of its opposing MGs. 

 Each MGO is typically fully aware of the presence of all $N$  MGs in the power grid and knows the size of their storage devices. In addition, each MGO knows the exact amount of energy excess available to its own MG. However, an MGO cannot determine the energy excess of other MGs. In fact, obtaining such information is not possible especially given the intermittent renewable energy sources and the time-varying nature of energy consumption. Each MGO thus assumes the excess of energy $Q_m$ of other MGs to be a random variable that follows a certain probability distribution function $f_n(Q_m)$ over $[0, Q_{m,\textrm{max}}]$ where $m \in \mathcal{N} \setminus \{n\}$. We refer to $Q_n$ as the \emph{type} of MGO $n$ and, to $f_n(Q_m)$, as MGO $n$'s \emph{belief} of another MGO $m$'s type. In fact, when MGO $n$ chooses a certain storage strategy $\alpha_n$, it is uncertain of the profit it will gain. This uncertainty stems from its incomplete information regarding the type of its opponents, originating from the intermittent renewable energy  and the time-varying nature of energy consumption,  as well as from randomness of an emergency event.

Given the competition over the financial incentives offered by the power company for emergency energy, the MGOs' actions and utility are highly interdependent thus motivating a game-theoretic approach \cite{walid}. In addition, given the incomplete information of the opponents' excess of energy that directly affects the MGOs' utility, each MGO will maximize its expected utility given its own beliefs $f_n(Q_m)$. MGO $n$'s expected utility, $E_n(\boldsymbol{\alpha},Q_n)$, will therefore be given by

\vspace{-0.4cm}
\begin{align}
E_n(\boldsymbol{\alpha},Q_n) = \mathbb{E}_{\boldsymbol{Q}_{-n}} \left[ U_n(\boldsymbol{\alpha}, \boldsymbol{Q})\right],
\end{align}
where  $\boldsymbol{Q}_{-n}$  is the vector that represents the energy excess of all MGs in the set  $\mathcal{N} \setminus \{n\}$.
The strategic interactions between the various MGOs under incomplete information can be modeled using Bayesian game models \cite{walid}.
\vspace{-0.2cm}

\subsection{Bayesian game formulation}
We formulate a static noncooperative Bayesian game \cite{walid} between the different MGOs in the set $\mathcal{N}$. In this game, each MGO seeks to maximize its expected utility given its beliefs of its opponents' energy excess by choosing its optimal storage strategy. Since the decisions on the portion of energy to store
are coupled, as captured by (2), we adopt a game-theoretic approach.
Formally, we define a strategic game 
$\Xi=\{\mathcal{N}, \{\mathcal{A}_n\}_{n\in\mathcal{N}}, \{\mathcal{T}_n \}_{n\in\mathcal{N}} , \{\mathcal{F}_n\}_{n\in\mathcal{N}} ,\{U_n\}_{n\in\mathcal{N}}\}$ where $\mathcal{N}$ is the set of all MGOs, $\mathcal{A}_n$ is the action space which represents the possible storage strategies of each player $n$, $\mathcal{T}_n$ is the set of types of MGOs that represent the  possible energy surplus for each their MGs,  $\mathcal{F}_n$ is the set of beliefs of player $n$ represented by the probability distributions of each of its opponents' types, and $U_n$ is the utility function of player $n$ defined in (\ref{eq:uti}).
In order to find the solution of the proposed game, we first define the two key concepts of  \emph{best response strategy} and \emph{Bayesian Nash equilibrium}.  
\begin{definition}
The set of \emph{best response strategies} of an MGO $n \in \mathcal{N}$  to the strategy profile $\boldsymbol{\alpha}_{-n}$, $r(\boldsymbol{\alpha}_{-n})$, is  defined as
\begin{flalign}\label{eq:brbr}
\small
\nonumber r_n(\boldsymbol{\alpha}_{-n})\!=\!\{\alpha_n^{*} \in \mathcal{A}_n| \mathbb{E}_{\boldsymbol{Q}_{-n}} \left[ U_n(\alpha^*_n,\boldsymbol{\alpha}_{-n}, \boldsymbol{Q}) \right] \geq \\
 \mathbb{E}_{\boldsymbol{Q}_{-n}} \left[ U_n(\alpha_n,\boldsymbol{\alpha}_{-n},  \boldsymbol{Q})  \right], \forall \alpha_n \in \mathcal{A}_n\},
\end{flalign}
 where  $\boldsymbol{\alpha}_{-n}$  is the vector that represents the storage strategy of all MGOs in the set  $\mathcal{N} \setminus \{n\}$.
\end{definition}
\normalsize
In other words, when the strategies of the opponents are fixed to $\boldsymbol{\alpha}_{-n}$, any best response strategy would maximize player $n$'s expected utility, given its beliefs $\mathcal{F}_n $ of its opponents' types. In our analysis, we assume that an MGO's belief $f_n(Q_m)$ over its opponent's  energy surplus  follows a uniform distribution over the domain $\left[0  , Q_{m,\textrm{max}}\right]$.
We next define the concept of a pure strategy Bayesian Nash equilibrium.

\begin{definition}
A strategy profile $\boldsymbol{\alpha}^*$ is said to be a \emph{pure strategy Bayesian  Nash equilibrium} if every MGO's strategy is a best response to the other
MGOs' strategies, i.e.
\begin{align}
\alpha_n^{*} \in  r_n(\boldsymbol{\alpha}^*_{-n}) \, \forall n \in \mathcal{N}.
\end{align}
\end{definition}\vspace{-0.4cm}
In the proposed game, at the BNE, no MGO $n$, can increase its expected utility by unilaterally deviating from its storage strategy $\alpha^{*}_n$.

 In what follows,  we will derive closed-form expressions of the BNEs for the case in which two MGs are located in the proximity of the critical load. In fact, power supply to the critical load from distant MGs might not be feasible due to transmission barriers  and significant power losses. As such, given these limitations and the scale of a given microgrid, the analysis for two MGs will be quite representative.



\section{Two-player Game solution under Classical Game Theory analysis}

For the case in which two MGs ($N=2$) are capable of supplying the critical load, the expected utility of MGO $1$ given its belief of MGO $2$'s type can be written as
\vspace{-0.0cm}
 \begin{align}
E_1(\boldsymbol{\alpha},Q_1) = \int_{0}^{Q_{2,\textrm{max}}}   U_1(\boldsymbol{\alpha},\boldsymbol{Q}) f_1(Q_2) dQ_2,
 \end{align}\\
where  $\boldsymbol{\alpha} = [ \alpha_1 \,\,  \alpha_2 ]$ and $\boldsymbol{Q} =  [ Q_1 \, \, Q_2 ]$. For the two-MG case, we have
\vspace{-0.0cm}\small
 \begin{align}\label{eq:U1}
    U_1(\boldsymbol{\alpha},\boldsymbol{Q}) =
    \begin{cases}
     \rho Q_1\left( 1 - \alpha_1  \right) + \theta \rho_{c} \alpha_1 Q_1  & \text{if}\    \alpha_2 \leq \frac{ L_c - \alpha_1 Q_1}  { Q_2}, \\
     \rho Q_1 \left( 1 - \alpha_1\right) + \theta \rho_{c} D_1  & \text{otherwise.}
    \end{cases}
 \end{align}
\normalsize
Next, we assume that neither of the MGs owns a large enough storage device to fully supply the critical load on its own. Under this assumption, $D_1$ will be given by

\small
\begin{align} 
D_1 =\alpha_1 Q_1 - \frac{1}{2} \left(\alpha_1 Q_1 +\alpha_2 Q_2 - L_c \right). 
\end{align}
\normalsize
In order to find the solution of the proposed game, we first derive the best response strategy of each player which we then use to compute the different BNEs. 
\normalsize
\vspace{-0.2cm}
\subsection{Derivation of the best response} 
The best response strategy of each  MGO is characterized next. In fact, we present the following propositions that analyze MGO $1$'s best response for different values of $\alpha_2$. 

\begin{proposition} 
The best response of MGO $1$, for $\alpha_2 \in \left[ 0 , \frac {L_c -  Q_1} {Q_{2,\textrm{max}}} \right]$, is given by  $r_{1}(\alpha_2) = 1$. MGO $1$  thus maximizes its expected utility by storing its MG's entire energy excess.

\end{proposition}

\vspace{-0.25cm}

\begin{proof}
For $\alpha_2 \leq \frac {L_c -  Q_1} {Q_{2,\textrm{max}}}$, the total stored energy is below the critical load for all types of MGO $2$ and all strategies of MGO $1$ since $\alpha_2 Q_{2,\textrm{max}} + Q_1 \leq L_c$. Thus, MGO $1$'s best response is to store its entire energy excess which is fully sold in case of emergency. In fact, here, \normalsize $E_{1}(\boldsymbol{\alpha},Q_1) = U_1(\boldsymbol{\alpha},\boldsymbol{Q}) =  \rho\left( Q_1 - \alpha_1 Q_1 \right) + \theta \rho_{c} \alpha_1 Q_1$ since $U_1(\boldsymbol{\alpha},\boldsymbol{Q})$ is independent of $Q_2$ for this case, as seen in  (\ref{eq:U1}). $E_1(\boldsymbol{\alpha},Q_1)$ is clearly an increasing function, given that $\rho_c\theta > \rho$, which is maximized at its upper boundary ($\alpha_1 = 1$). Thus $ r_{1}(\alpha_2) =1$  for $\alpha_2 \in \left[ 0 , \frac {L_c -  Q_1} {Q_{2,\textrm{max}}} \right]$.
\end{proof}
\vspace{-0.4cm}
\begin{proposition} 
The best response of MGO $1$, for $\alpha_2 \in \left[\frac {L_c -  Q_1} {Q_{2,\textrm{max}}}, 1 \right]$, is given by \\
\begin{align}\label{eq:brbrbrbr} \footnotesize  
   r_{1}(\alpha_2)  =
   \begin{cases}
    \frac {L_c \rho_c \theta + ( \rho_c \theta - 2\rho)\alpha_2 Q_{2,\textrm{max}} } { Q_1 \rho_c \theta}, & \text{if}\   \, \left[\frac {2\rho}{\rho_c\theta}-1\right] \alpha_2 > \frac{L_c -  Q_1} {Q_{2,\textrm{max}}},  \\
    1, & \text{if}  \, \left[\frac {2\rho}{\rho_c\theta}-1\right] \alpha_2 \leq \frac{L_c -  Q_1} {Q_{2,\textrm{max}}}.\\
    \end{cases}
\end{align}
\end{proposition}

\normalsize
\vspace{-0.3cm}
\begin{proof}
The proof of this proposition is in Appendix A.
\end{proof}

Given the previous propositions, an MGO's  best response strategy  is thus summarized in the following theorem.

\begin{theorem}
 The best response strategy of MGO $1$, $r_{1}(\alpha_2)$, is given by
\begin{align}\label{eq:brbrbr} \small
 r_{1}(\alpha_2) =
   \begin{cases}
   1, & \text{if}\   \alpha_2 \leq \frac {L_c -  Q_1} {Q_{2,\textrm{max}}}, \\
    \alpha_{1,r}, & \text{if}\  \alpha_2 > \frac {L_c -  Q_1} {Q_{2,\textrm{max}}} \, \textrm{and}  \, \left[\frac {2\rho}{\rho_c\theta}-1\right] \alpha_2 > \frac {L_c -  Q_1} {Q_{2,\textrm{max}}},  \\
    1, & \text{if}\ \alpha_2 > \frac {L_c -  Q_1} {Q_{2,\textrm{max}}} \, \textrm{and}  \, \left[\frac {2\rho}{\rho_c\theta}-1\right] \alpha_2 \leq \frac{L_c -  Q_1} {Q_{2,\textrm{max}}}.\\
    \end{cases}
\end{align}
\normalsize
\noindent MGO $2$'s best response strategy $r_2(\alpha_1)$ is derived similarly and is the same as (\ref{eq:brbrbr}) but with indices $1$ and $2$ interchanged.
\end{theorem}

\begin{proof}
The proof follows  from Propositions 1 and 2.
\end{proof}

\subsection {Derivation and interpretation of the equilibria}

Given the MGOs' best response function in (\ref{eq:brbrbr}), we will compute all possible BNEs for this game. We will then derive and interpret the conditions needed for each BNE to exist.

\begin{theorem}
The proposed  MGO game admits four possible  Bayesian Nash equilibria for different conditions that relate the MG parameters, $Q_n$ and $Q_{n,\textrm{max}}$, with power grid parameters $\rho, \rho_c, \theta$, and $L_c$. The strategy profiles ($\alpha_1^*,\alpha_2^*$), that constitute the four BNEs, are the following:\\ \normalsize
1) First BNE: (1,1).\\
2) Second BNE:   $\left(1, \dfrac {L_c \rho_c \theta + ( \rho_c \theta - 2p) Q_{1,\textrm{max}}} { Q_2 \rho_c \theta} \right)$.\\ \normalsize
3) Third BNE:  $\left(  \dfrac {L_c  \rho_c  \theta + ( \rho_c  \theta - 2p) Q_{2,\textrm{max}}} { Q_1 p_c \theta},\normalsize 1 \right)$. \\ \normalsize 
4) Fourth BNE: $\left(\alpha^{*}_{1,4} , \alpha^{*}_{2,4}\right)$ is the strategy profile that constitute the fourth BNE, where \\\\ 
$\alpha^{*}_{1,4} = \dfrac {- L  \rho_c  \theta (Q_2 \rho_c \theta - 2 Q_{2,\textrm{max}}\rho  + Q_{2,\textrm{max}}\rho_c \theta )} {Q_{1,\textrm{max}} Q_{2,\textrm{max}} \left( 4\rho^2 +  \rho_c^2\theta^2 - 4  \rho \rho_c \theta \right) - Q_1 Q_2 \rho_c^2  \theta^2 }$\normalsize,\\ \vspace{0.15cm}
$\alpha^{*}_{2,4} = \dfrac {- L  \rho_c  \theta (Q_1 \rho_c \theta - 2Q_{\textrm{max},1}\rho  + Q_{1,\textrm{max}}\rho_c \theta )} {Q_{1,\textrm{max}} Q_{2,\textrm{max}} \left( 4\rho^2 +  \rho_c^2\theta^2 - 4  \rho \rho_c \theta \right) - Q_1 Q_2 \rho_c^2  \theta^2 }$. \\
\normalsize
\end{theorem}

\vspace{-0.3cm}
\begin{proof}
The strategy profiles of the BNEs are derived by solving the set of best-response equations, $\alpha_1^* = r_1(\alpha_2^*)$ and $\alpha_2^*=r_2(\alpha_1^*)$,
for the different possible combinations of the best response strategies.
\end{proof}

\vspace{-0.1cm}
\noindent The conditions under which each BNE is defined are further summarized and interpreted next.\\

\vspace{-0.1cm}

\subsubsection {First BNE} the strategy profile (1,1) constitutes a BNE of the proposed game if any of the following four conditions is satisfied:\\
\vspace{-0.4cm}

 a) \, \normalsize $L_c \geq Q_{2,\textrm{max}} + Q_1$  and $L_c \geq Q_{1,\textrm{max}} + Q_2$. Here, each MGO is aware that the total stored energy is below the critical load, regardless of the type and strategy of its opponent.

 b)\, \normalsize $L_c \geq Q_{2,\textrm{max}} + Q_1$  and  $ \frac {2\rho}{\rho_c\theta}-1   \leq \frac{L_c -  Q_2} {Q_{1,\textrm{max}}} <1 $.  \normalsize Here, MGO $1$ knows that the total stored energy is always below the critical load regardless of the type and strategy of its opponent. On the other hand, MGO $2$ is aware that part of its MG's stored energy might not be sold in case of emergency. However, $\rho_c$ is large enough compared to $\rho$ to satisfy the condition under which MGO $2$ stores its MG's entire excess.

 c)\, \normalsize $\frac {2\rho}{\rho_c\theta}-1   \leq \frac{L_c -  Q_1} {Q_{2,\textrm{max}}} <1 $ and $L_c \geq Q_{1,\textrm{max}} + Q_2$. \normalsize The analysis of this condition is the same as condition b) with the order of the players reversed.  

d)\,\normalsize $\frac {2\rho}{\rho_c\theta}-1   \leq \frac{L_c -  Q_1} {Q_{2,\textrm{max}}} <1 $ and $ \frac {2\rho}{\rho_c\theta}-1   \leq \frac {L_c -  Q_2} {Q_{1,\textrm{max}}} <1 $. \normalsize In this case, both MGOs are aware that part of their stored energy might not be sold. However, $\rho_c$ is large enough compared to $\rho$ to satisfy the conditions for which both MGOs store their MGs' entire excess.
\normalsize

\vspace{-0.05cm}

\subsubsection {Second BNE}  the strategy profile \normalsize $\left(1, \frac {L_c \rho_c \theta + ( \rho_c \theta - 2p) Q_{1,\textrm{max}}} { Q_2 \rho_c \theta} \right)$ \normalsize constitutes a BNE of the proposed game if any of the following two conditions are satisfied: \vspace{0.1cm}

a)  \vspace{0.1cm} \normalsize $L_c \tiny \geq  \frac {L_c \rho_c \theta + ( \rho_c \theta - 2p) Q_{1,\textrm{max}}} { Q_2 \rho_c \theta} Q_{2,\textrm{max}} + Q_1$ \normalsize  and \\
\normalsize  $\frac {2\rho}{\rho_c\theta}-1  > \frac {L_c -  Q_2} {Q_{1,\textrm{max}}} $. \normalsize
In this case, MGO $1$ knows that given MGO $2$'s storage strategy, the total stored energy is always below the critical load. Meanwhile, MGO $2$ is aware that, given MGO $1$'s strategy, the total stored energy might exceed the critical load and part of its stored energy might not be sold in case of emergency. MGO $2$ will not store the entire excess given that $\rho_c$ is not large enough compared to $\rho$. \vspace{0.1cm}

b)  $ \left[\frac {2\rho}{\rho_c\theta}-1\right] \frac {L_c \rho_c \theta + ( \rho_c \theta - 2\rho) Q_{1,\textrm{max}}} { Q_2 \rho_c \theta} \leq \frac{L_c -  Q_1} {Q_{2,\textrm{max}}}$\normalsize,

\noindent  $\frac {L_c -  Q_1} {Q_{2,\textrm{max}}}    <\frac {L_c \rho_c \theta + ( \rho_c \theta - 2\rho) Q_{1,\textrm{max}}} { Q_2 \rho_c \theta} $   and   $ \frac {2\rho}{\rho_c\theta}-1  > \frac{L_c -  Q_2} {Q_{1,\textrm{max}}}$\normalsize. 
Here, both MGOs know that given their opponent's strategy, part of their MG's stored energy might not be sold. The emergency price  $\rho_c$ is large enough compared to $\rho$ to satisfy the condition for which MGO $1$ stores the entire excess, however, it is not large enough for MG $2$ to fully store its MG's entire excess.

\normalsize
\vspace{5 mm}
\vspace{-0.2cm}
\subsubsection {Third BNE} The interpretation of the third BNE is similar to that of the second but with index 1 swapped with 2.
\\
\subsubsection{Fourth BNE} The strategy profile
$\left(\alpha^{*}_{1,4} , \alpha^{*}_{2,4}\right)$, defined in Theorem 2,  constitutes a BNE which is obtained by solving the set of equations $\alpha_1^{*} =  \alpha_{1,r}$ \footnotesize and \normalsize   $\alpha_2^*=  \alpha_{2,r}$, in the case where the following condition is satisfied: \\

 a)  $\alpha^{*}_{2,4} \left[\frac {2\rho}{\rho_c\theta}-1\right] >  \frac{L_c -  Q_1} {Q_{2,\textrm{max}}}$ 
  and $\alpha^{*}_{1,4}  \left[\frac {2\rho}{\rho_c\theta}-1\right] >\frac {L_c -  Q_2} {Q_{1,\textrm{max}}} $. \\
\vspace{-0.1cm}

\normalsize \noindent Under this condition, both MGOs know that given their opponent's strategy, part of their MG's stored energy might not be sold. The emergency price $\rho_c$  is not large enough to satisfy the conditions under which either MGO stores the entire excess. 

Our previous analysis assumes that all MG operators are fully rational and their behavior can thus be modeled using classical game-theoretic analysis. However, this assumption might not hold true in a real smart grid, given that the operators of the MGs might have different subjective valuations of the payoffs gained from selling their energy surplus.  Next, we will use the framework of prospect theory \cite{nobel} to model the behavior of MGOs when faced with such uncertainty and subjectivity of profits, stemming from the presence of renewable energy and the uncertainty it imposes on the volume of energy surplus that other MGOs generate. 

\vspace{-0.1cm}
\section{ prospect  theoretic analysis}

 \vspace{-0.05cm}
In a classical noncooperative game, a player evaluates an objective expected utility. However, in practice, individuals tend to subjectively perceive their utility when faced with  uncertainty \cite{nobel}. In our model, an MGO's uncertainty originates from the presence of renewable energy and the uncertainty it imposes on the volume of energy surplus that the opposing MGOs generate. In fact, an MGO is uncertain of the portion of its MG's stored energy that will be sold in case of emergency, which is directly related to the energy surplus available to its opponents. Since MGOs are humans, they will perceive the possible profits of energy trading, in terms of gains and losses.

This motivates the application of PT to account for the MGO's subjectivity while choosing the optimal energy portion to store. PT is a widely used tool for understanding  human behavior when faced with uncertainty of alternatives. In our analysis, we will inspect the effect of the key notion of  utility framing from prospect theory. Utility framing states that a utility is considered a gain if it is larger than the reference point, while it is perceived as a loss if it is smaller than that reference point. We define $R_n$ as the reference point of a given MGO $n$. The choice of $R_n$ can be different between MGOs as it reflects personal expectations of profit from selling the energy surplus. In this regard, a certain profit, $r$,  originating from a particular energy trade, will be perceived differently by an MGO used to reaping larger profits as opposed to an MGO that usually generates lower profits. In fact, an MGO $n$ with historically high profits would have a high reference point, $R_n>r$, and will hence consider $r$ to be a loss, whereas, an MGO $m$ with relatively low historical profits  would have a lower reference point, $R_m<r$ and would hence consider $r$ to be a gain. Consequently, to model this subjective perception of losses and gains we need to redefine the utility function of the MGOs using PT framing \cite{CPT}:
 \vspace{-0.5cm}

 \begin{equation}\label{eq:rule}
\footnotesize
 V\left(U_n\left(\boldsymbol{\alpha},\boldsymbol{Q}\right) \right)=\begin{cases}
\left( U_n(\boldsymbol{\alpha},\boldsymbol{Q}) -R_n\right)^{\beta^+} &\textrm{ if } U_n(\boldsymbol{\alpha},\boldsymbol{Q}) > R_n, \\
-\lambda_n  \left( R_n -  U_n(\boldsymbol{\alpha},\boldsymbol{Q})\right)^{\beta^-} &\textrm { if } U_n(\boldsymbol{\alpha},\boldsymbol{Q}) < R_n,
\end{cases}\\
 \end{equation}
\normalsize
where $0<{\beta^-}\leq 1 , 0<{\beta^+} \leq 1$ and $\lambda \geq 1$.

\normalsize

 $V(\cdot)$ is the framing value function that is concave in gains and convex in losses with a larger slope for losses than for gains \cite{CPT}. In fact, PT studies show that the aggravation that an individual feels for losing a sum of money is greater than the satisfaction associated with gaining the same amount\cite{nobel}, which explains the introduction of the loss multiplier $\lambda_n$. In addition, the framing principle states that an individual's sensitivity  to marginal change in its utility diminishes as we move further away from the reference point, which explains the introduction of the gain and loss exponents $\beta^+$ and $\beta^-$.
%

%


It is important to note that, as an MGO chooses to store a larger portion $\alpha$ of its MG's energy,  its potential payoffs will now span a larger range of  values. In other words, as an MGO stores more energy, it will now have the possibility to make higher expected profits by selling more in case of emergency. On the other hand, by storing more energy, the MGO risks making less profit whenever its opponent has also stored a significant part of its own energy. These probable payoffs are related to the type of the opponent. In fact, the MGO would get a maximum profit for the case in which the opponent's type is small, i.e. the opponent did not have a significant energy surplus. For the case in which the opponent's type is large, a significant part of an MG's stored energy will not be sold in case of emergency, resulting in lower possible payoffs for its MGO, compared to smaller values of $\alpha$. This concept is key in our  PT analysis, given that payoffs are evaluated through comparison to the reference point. 
Similarly to our analysis for the CGT case, we will first derive the best response strategy of the MGOs.

\begin{proposition}
The best response of MGO $1$ under PT, for $\alpha_2 \in \left[ 0 , \frac {L_c -  Q_1} {Q_{2,\textrm{max}}} \right]$, is to store its entire energy excess, similarly to the classical game theory analysis.
\end{proposition} 

\begin{proof}
As seen from Proposition 1,  for $\alpha_2 \in \left[ 0 , \frac {L_c -  Q_1} {Q_{2,\textrm{max}}} \right]$, $U_1(\boldsymbol{\alpha},\boldsymbol{Q}$) is an increasing function over its domain. Given that the framing function $V(\cdot)$ is an increasing function as well, MGO $1$'s expected utility, $E_{1,\textrm{PT}}(\boldsymbol{\alpha},Q_1) =V\left(U_1\left(\boldsymbol{\alpha}, \boldsymbol{Q}\right)\right)$, is thus maximized at its upper boundary of $\alpha_1 = 1$.
\end{proof}

\vspace{-0.2cm}
We next derive the expected utility of MGO $1$ under PT for  $\alpha_2 \in \left[ \frac {L_c -  Q_1}  {Q_{2,\textrm{max}}} , 1 \right]$. MG $1$'s expected utility for $\alpha_2 \in \left[ \frac {L_c -  Q_1}  {Q_{2,\textrm{max}}} , 1 \right]$ takes different values for $\alpha_1 \in \left[0 , \frac {L_c -  \alpha_2 Q_{2,\textrm{max}} } {Q_1}\right]$ and $ \alpha_1 \in \left[\frac {L_c -  \alpha_2 Q_{2,\textrm{max}}}{Q_1},1\right]$:

\begin{proposition}
 For $\alpha_1 \in \left[0 , \frac {L_c -  \alpha_2 Q_{2,\textrm{max}} }  {Q_1}\right] $ and $\alpha_2 \in \left[ \frac {L_c -  Q_1}  {Q_{2,\textrm{max}}} , 1 \right]$, MGO 1's expected utility under PT, $E_{\textrm{PT},1,2a}$,  is given by

\vspace{-0.6cm}
 \begin{equation}\label{eq:cpt2a}\small
  E_{\textrm{PT},1,2a}(\boldsymbol{\alpha},Q_1) =\begin{cases}
 -\lambda_1  \left( R_1-U_{1,2a} \right)^{\beta_1^-}     &\textrm{ if } \alpha_1 \leq B \normalsize , \\
  \left(U_{1,2a} - R_1 \right) ^{\beta_1^+}&\textrm { if }  \alpha_1 > B\normalsize,
\end{cases}
 \end{equation}
where \small$U_{1,2a} = \rho\left( Q_1 - \alpha_1 Q_1 \right) - \theta \rho_{c} \alpha_1 Q_1$, and $B=\frac {R_1- \rho Q_1}  {Q_1\left(\rho_c \theta - \rho \right)}$. 
\end{proposition}
\begin{proof}
In Proposition 4, Equation (\ref{eq:cpt2a}) follows from the fact that for {$\alpha_1 \leq B$}, \normalsize the original utility, $U_{1,2a}$, is below the reference point $R_1$ and is thus perceived as a loss. On the other hand, it is considered as a gain for $\alpha_1 > B$.
\end{proof}

\begin{proposition}
 For $\alpha_1 \in \left[ \frac {L_c -  \alpha_2 Q_{2,\textrm{max}} }  {Q_1}, 1\right] $ and $\alpha_2 \in \left[ \frac {L_c -  Q_1}  {Q_{2,\textrm{max}}} , 1 \right]$, player 1's expected utility under PT is given by
\begin{align}
E_{\textrm{PT},1,2b}(\boldsymbol{\alpha},Q_1) = I_1 + I_2,
\end{align}
\normalsize
\vspace{-0.6cm}

where 
 \begin{equation}\label{eq:CPI1}
\footnotesize
I_1 =\begin{cases} \vspace{0.1cm}
 -\dfrac{ \lambda_1 (L_c - \alpha_1Q_1)}{\alpha_2Q_{\textrm{max},2}}  \left[ R_1- U_{I,1}\right]^{\beta_1^-}     &\textrm{ if } \alpha_1 \leq B, \\
  \dfrac{ L_c - \alpha_1Q_1 }{\alpha_2Q_{\textrm{max},2}} \left[ U_{I,1} - R_1  \right] ^{\beta_1^+}&\textrm { if }  \alpha_1 > B,
\end{cases}
\normalsize
 \end{equation}

\vspace{-0.2cm}
\begin{equation} 
U_{I,1} = \rho\left( Q_1 - \alpha_1 Q_1 \right) + \theta \rho_{c} \alpha_1 Q_1,
\end{equation}

\vspace{-0.5cm}
\begin{equation}\label{eq:CPI2} 
\footnotesize \vspace{-0.15cm}
 I_2=\begin{cases}
M_l \left[ \left(  R_1-U_{\textrm{max},2} \right)^{\beta_1^- +1} - 
\left( R_1-U_{A,2} \right)^{\beta_1^- +1} \right]  &\textrm{ if }  C_1 , 
\\ 
M_g \left[ \left( U_{r,2} - R_1 \right)^{\beta_1^+ +1} -  \vspace{-0.15cm} \vspace{-0.15cm}
\left(U_{A,2} - R_1  \right)^{\beta_1^+ +1} \right] +\\ \vspace{0.0cm}
M_l \left[ \left(  R_1 - U_{\textrm{max},2}  \right)^{\beta_1^- +1} - 
\left(  R_1- U_{r,2}\right)^{\beta_1^- +1} \right]  &\textrm{ if } {C}_2\footnotesize, \vspace{-0.15cm}\vspace{-0.15cm}
\\ \\ 
M_g \left[ \left( U_{\textrm{max},2} - R_1  \right)^{\beta_1^+ +1} - 
\left(U_{A,2} - R_1  \right)^{\beta_1^+ +1} \right]  &\textrm{ if } {C}_3 \footnotesize,
\end{cases}
\end{equation}\small
 $M_g =  \dfrac {-2}{\left(\beta_1^+ +1\right) \rho_c\theta \alpha_2}$, $M_l = \dfrac {-2\lambda_1}{\left(\beta_1^- +1\right) \rho_c\theta \alpha_2}$, $U_{\textrm{max},2} = \rho\left( Q_1 - \alpha_1 Q_1\right) + \frac {1}{2} \theta \rho_{c} \left(\alpha_1 Q_1 + L_c  -  Q_{2,\textrm{max}}\right)$,
$U_{\textrm{A},2} = \rho\left( Q_1 - \alpha_1 Q_1\right) + \frac {1}{2} \theta \rho_{c} \left(\alpha_1 Q_1 + L_c  -  A \right)$, $A =  \frac {L_c - \alpha_1 Q_1}  {\alpha_2 }$,
and $U_{r,2} = \rho\left( Q_1 - \alpha_1 Q_1\right) + \frac {1}{2} \theta \rho_{c} \left(\alpha_1 Q_1 + L_c  -  Q_{2,r}\right)$. \normalsize $Q_{2,r}$ is given in (\ref{eq:Q_r_2}). \\
Condition $C_1$, $C_2$, and $C_3$ are given by \small
\begin{align}\label{eq:condition_1}
C_1: \alpha_1 \leq B, \hspace{5.755cm}
\end{align}
\vspace{-0.7cm}
\small
\begin{multline}\label{eq:condition_2} 
C_2: \alpha_1 > B \,\,\, \textrm{and} \,\,\, \\
Q_1 \left(  \theta\rho_c - 2\rho \right) \alpha_1 \leq  \theta\rho_c\alpha_2Q_{\textrm{max},2} -  L_c\rho_c\theta - 2\rho Q_1 + 2R_1,
 \end{multline}
\vspace{-0.7cm}
\noindent
 \begin{multline}\label{eq:condition_3} 
C_3: \alpha_1 > B \,\,\, \textrm{and} \,\,\, \\
Q_1 \left(  \theta\rho_c - 2\rho \right)   \alpha_1 >  \theta\rho_c\alpha_2Q_{\textrm{max},2} -  L_c\rho_c\theta - 2\rho Q_1 + 2R_1.
 \end{multline}
\normalsize
MGO 2's expected utility function is derived in a similar manner as MGO 1's with indices 1 and 2 reversed.

\end{proposition}
\begin{proof}
The proof is given in Appendix C.
\end{proof}

%
Given the complex structure of each MGO's expected utility function with framing, computing the closed-form expression of the best response strategy is difficult for the PT case. In particular,  the analysis of $E_{\textrm{PT},1,2b}$ is quite challenging due to the various forms that the function can take under different conditions as seen in (\ref{eq:CPI1}) and (\ref{eq:CPI2}). Therefore, in order to find the BNE under PT, a best response algorithm is proposed.

This iterative algorithm dictates that, in response to its opponent's current strategy, each MGO sequentially chooses its optimal storage strategy by numerically characterizing, from its action space, the action that maximizes its expected utility. In fact, given the closed-form expressions provided in Propositions 3, 4, and 5, an MGO can easily compute its expected utility for each of its strategies. 
%
In this respect, upon convergence, this algorithm is guaranteed to reach an equilibrium \cite{walid}. In fact, at the point of convergence, each MGO is playing the strategy that maximizes its expected PT utility facing its opponent's strategy. Hence, the MGOs will reach a BNE from which none has any incentive to deviate since such deviation would not improve their expected payoff. Indeed, as observed in our simulations in Section V, the algorithm always converged to an equilibrium. 
\vspace{-0.2cm}
%
%

\section{Simulation Results and analysis}

For our simulations,  we consider a smart grid with $N=2$ MGs capable of supplying power to one of the power grid's critical loads which requires a total of $L_c = 200$ kWh to remain operational until regular power supply is restored. We also assume the regular price per unit of energy to be $\rho = \$0.1$ per kWh. In addition, we take $\theta = 0.01$, and $\rho_c = \$11.6$ per kWh  unless stated otherwise.  The exponents $\beta^+$ and $\beta^-$ are taken to be both equal to 0.88 and the loss multiplier $\lambda=2.25$ unless stated otherwise  \cite{CPT}.  We simulate the system for two scenarios: CGT, and PT under utility framing.


\begin{figure}[t!]\label{fig:1}
  \begin{center}
   \vspace{-0.1cm}    \includegraphics[width=1\linewidth,height=.15\textheight]{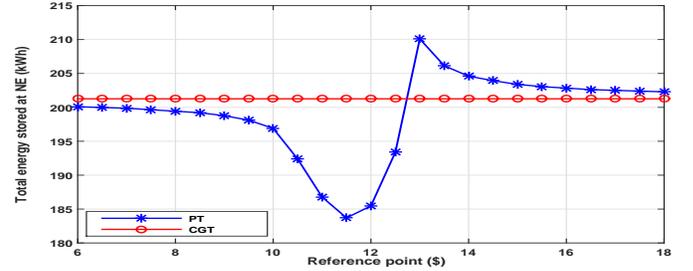}
    \vspace{-0.8cm}
  \caption{\label{fig:1}Total stored energy under classical game theory and prospect theory.}
  \end{center}\vspace{-0.5cm}
\end{figure}

\begin{figure}[t!]\label{fig:2}
  \begin{center}
   \vspace{-0.1cm}    \includegraphics[width=1\linewidth,height=.15\textheight]{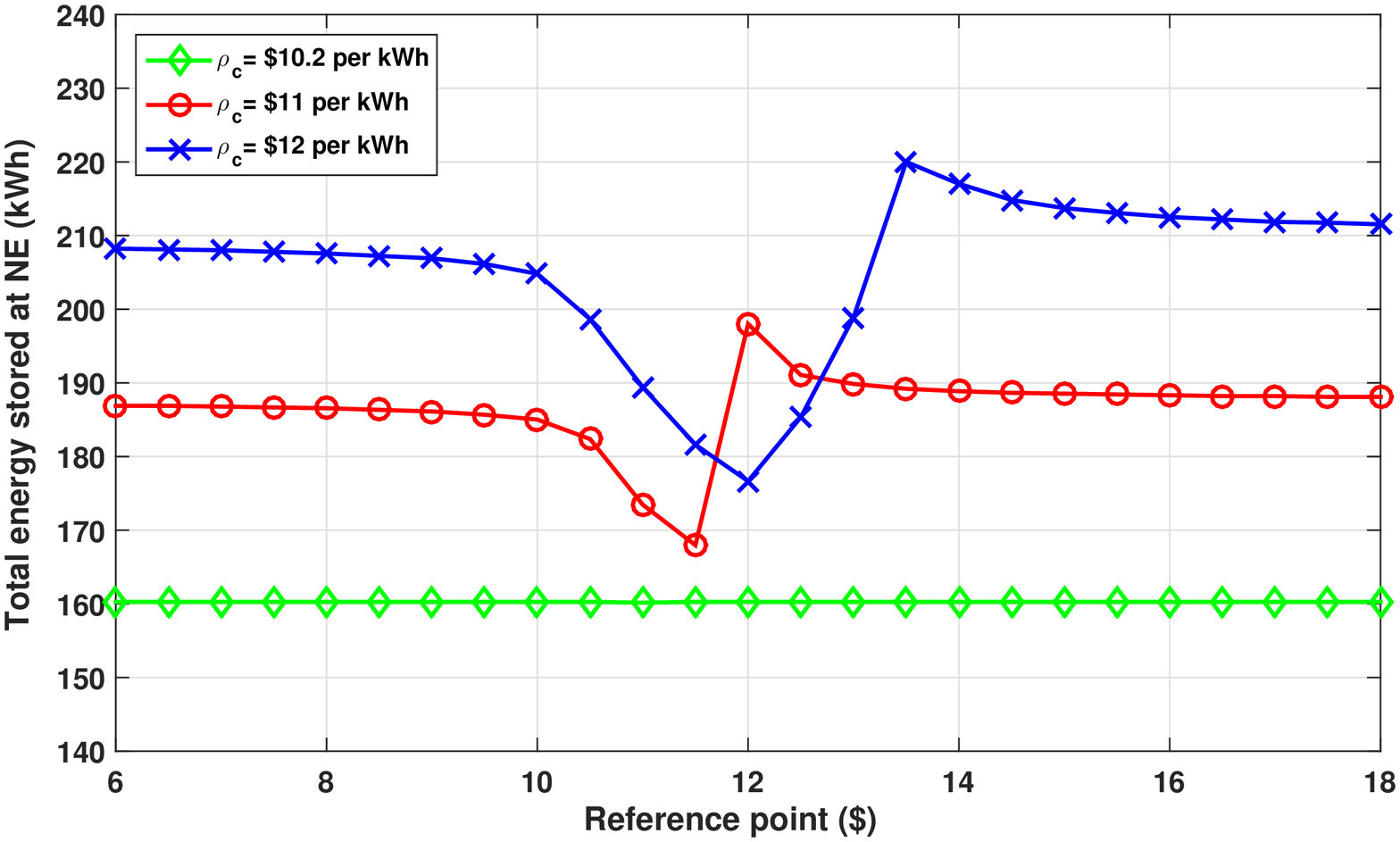}
    \vspace{-0.8cm}
  \caption{\label{fig:2}Effect of emergency price on PT sensitivity to the reference point.}
  \end{center}\vspace{-0.8cm}
\end{figure}

\indent Fig. \ref{fig:1} compares the effects of different MGO  reference points on the total  energy stored for both CGT and PT analysis. In the classical game theory case $ ( \beta^+ = \beta^- = \lambda =1  )$, an MGO's reference point is irrelevant given that losses and gains are computed in an identical objective manner. For the PT case, for a reference point below $\$8$, the BNE action profile is not significantly affected compared to the classical game theory case, since most potential payoffs of the BNE actions are still viewed as gains above the reference point. As the reference point increases from $\$8$ to $\$11.5$, the total 
stored energy will decrease from around $200$ to $184$ kWh, since some of the potential payoffs of the current BNE will start to be perceived as losses, as they cross the reference point. Given that losses have a larger weight under PT compared to classical game theory, the expected utility of the current strategy profile will significantly decrease, thus causing the BNE to drift towards lower storage strategies. The MGOs will exhibit risk averse behavior as they sell more of their energy at the current risk-free retail market price $\rho$. In fact, as previously mentioned, by decreasing $\alpha$, the minimum potential payoffs are larger, compared to the larger values of $\alpha$, and are still above the reference point.

The described behavior is reversed in the $\left[11.5 , 13\right]$ range where the MGOs start exhibiting more risk seeking behavior, i.e., storing more energy, to reach a total stored energy of $210$ kWh. In fact, the low risk strategies' potential payoffs are now fully perceived as losses causing a significant devaluation of their expected utility values. The BNE will thus go towards higher values of $\alpha$ with larger maximum payoffs, compared to lower values of $\alpha$, which are partially still considered as gains. Finally, when the reference point is above $\$13.5$, most potential payoffs of most strategies are now perceived as losses and the effect of PT will diminish gradually, and the total energy stored will reach $202$ kWh, identically to classical game theory. It is important to note that the critical load energy requirements are 200 kWh, which is met with the stored energy of the MGs under classical game theory but not necessarily under PT analysis. This highlights the need for an accurate behavioral analysis of the studied system. 

\indent Fig. \ref{fig:2} shows the effect of changing the emergency price $\rho_c$ on the role of the reference point in an MGO's decision, for $\lambda = 4$. For a price of $\rho_c = \$10.2$ per kWh, the total energy stored does not vary with the reference point. In fact, the expected future profits gained from storing energy are close to the profits incurred by selling at the current market price. On the other hand, when the price is increased to $\rho_c = \$11$ per kWh, the total stored energy will vary with the reference point by up to $10 \%$ from its original value. In fact, storing energy will now yield significantly higher expected future profits,  compared to selling at the current market price. Thus, an MGO's risk-seeking or risk-averse behavior is justified given the increasing uncertainty in profits.  Similarly, when $\rho_c=\$12$ per kWh, the total stored energy would vary further with the changing reference point,  by up to  $17\%$ from its original value.

\begin{figure}[t!]\label{fig:3}
  \begin{center}
   \vspace{-0.1cm}    \includegraphics[width=1\linewidth,height=.15\textheight]{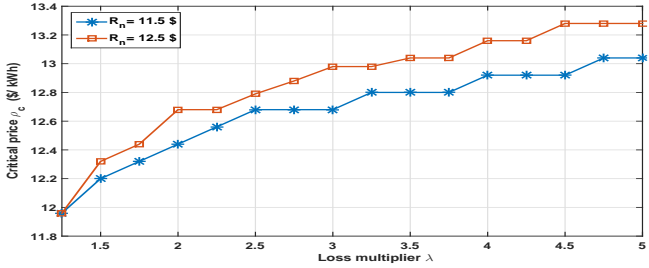}
    \vspace{-0.8cm}
  \caption{\label{fig:3}Emergency price needed to cover $L_c$ as a function of $\lambda$.}
  \end{center}\vspace{-0.5cm}
\end{figure}

\indent Fig. \ref{fig:3} shows the effect of the loss multiplier $\lambda$ on the  emergency price $\rho_c$ needed to cover the critical load for the reference points of $\$11.5$ and $\$12.5$. The effect of framing is more prominent as the loss multiplier increases. In fact, the MGOs will exhibit more risk averse behavior for the specified reference points as $\lambda$ increases, thus prompting the power company to increase the critical price in order to cover the critical load. In fact, as $\lambda$ increases, so will the valuation of the MGOs' losses. To avoid the large losses, the MGOs will decrease the energy stored by their MGs and will tend to sell more energy at the current risk free market price. This highlights the importance of behavioral analysis in choosing the proper pricing mechanism in smart grid resilience planning.

\begin{figure}[t!]\label{fig:4}
  \begin{center}
   \vspace{-0.1cm}    \includegraphics[width=1\linewidth,height=.15\textheight]{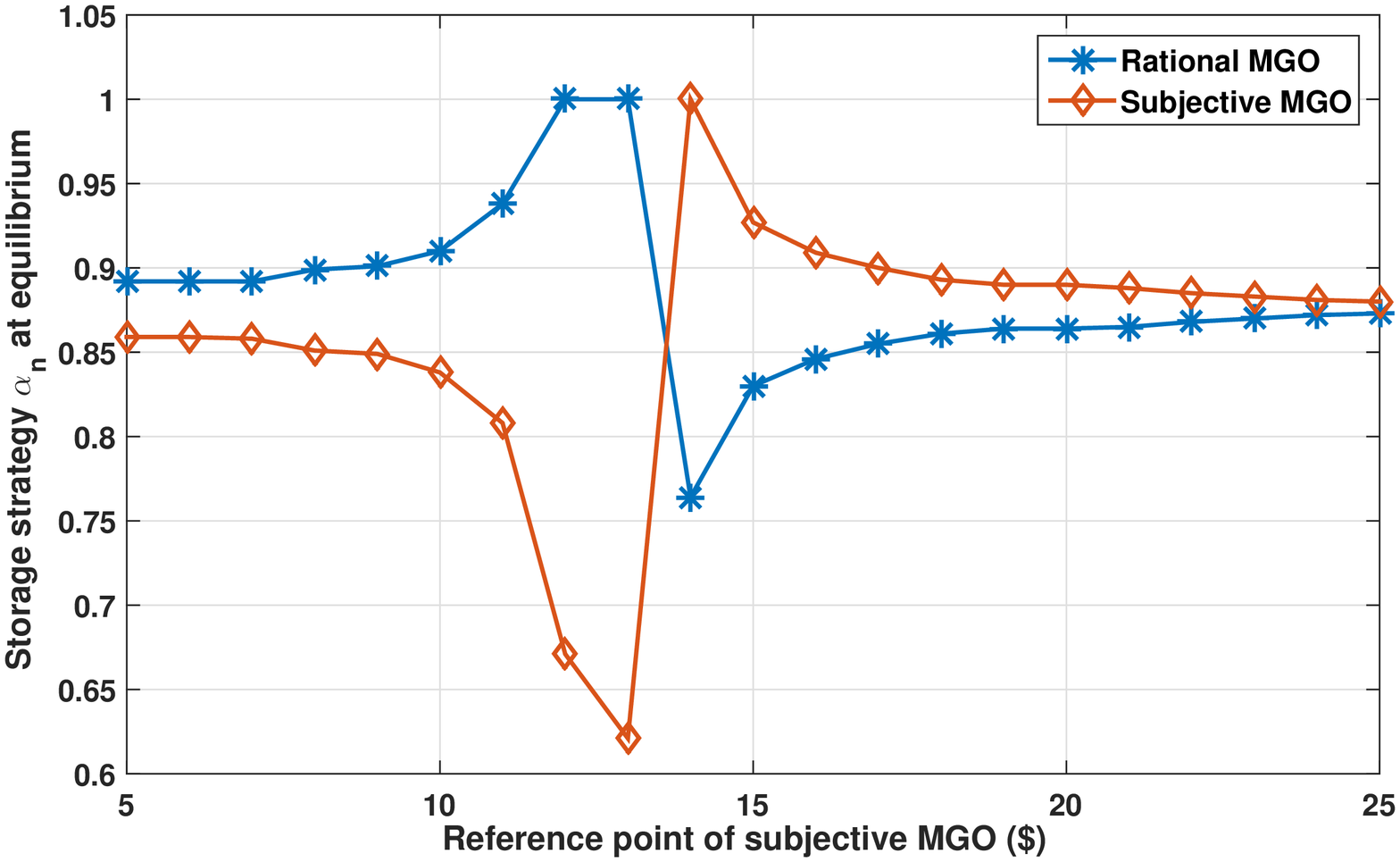}
    \vspace{-0.8cm}
  \caption{\label{fig:4}Storage strategies at equilibrium for a case with one subjective (PT) MGO and one rational (CGT) MGO.}
  \end{center}\vspace{-0.8cm}
\end{figure}
\vspace{-0.1cm}

\indent Fig. \ref{fig:4} illustrates the storage strategies at equilibrium for the case in which one of the MGOs is fully rational, while the second is subjective. The rational MGO will naturally have no reference point. Here, both MGs have the same size of storage $Q_{\textrm{max}}=150 \, \textrm{kWh}$ and energy excess available $Q = 120\, \textrm{kWh}$. As seen in  Fig. \ref{fig:4}, as the reference point of the subjective MGO increases from $\$5$ to $\$13$, it will exhibit risk averse behavior and decrease the portion of energy it stores, to reach a value of $0.625$. This is similar to the analysis of Fig. \ref{fig:1}. To respond, the rational MGO will hence increase the portion of energy stored to reach its maximum of $1$, given the lower stored energy of its opponent. As the reference point increases from $\$13$ to $\$14.5$, the subjective MGO will exhibit more risk seeking behavior and increase the portion of energy stored to reach its maximum of $1$. The rational MGO, will thus decrease its MG's stored energy,  given the storage strategy of its opponent. Finally, as the reference point increases from $\$14.5$ to $\$25$, the effect of utility framing will gradually decrease, and the storage strategy of both MGOs will reach a value of $0.88$. Given the negligible effect of PT at the high reference point of $\$25$, both MGOs, rational and subjective, will have equal strategies at equilibrium and thus similar behavioral patterns.
\vspace{-0.2cm}

\section{Conclusion}
\vspace{-0.1cm}
In this paper, we have proposed a novel framework for analyzing the storage strategy of micorgrid operators in an attempt to enhance smart grid resilience. We have formulated the problem as a Bayesian game  between multiple MGOs, who must choose the portion of their microgrids' excess to store, in order to maximize their expected profits. The MGOs play a noncooperative game, which is shown to have four Bayesian Nash equilibria  for the two MG case, under different conditions. Subsequently, we have used the novel concept of utility framing from prospect theory to model the behavior of MGOs when faced with the uncertainty of their opponents' energy surplus. Simulation results have highlighted the impact of behavioral considerations on the overall process of enhancing the resilience of  a smart grid by exploiting distributed,  microgrid energy storage.
\vspace{-0.1cm}

\def\baselinestretch{0.8}
\bibliographystyle{IEEEtran}
\bibliography{references}

\begin{appendices}
\section{Proof of Proposition 2}
\linespread{0.9}

For the proof of Proposition $2$, first, we analyze the expected utility of MGO $1$, for $\alpha_1 \in \left[0 , \frac {L_c -  \alpha_2 Q_{2,\textrm{max}} } {Q_1}\right]$ and $ \alpha_1 \in \left[\frac {L_c -  \alpha_2 Q_{2,\textrm{max}}}{Q_1},1\right]$, with $\alpha_2 \in \left[\frac {L_c -  Q_1} {Q_{2,\textrm{max}}}, 1 \right]$.
\\
\\
\indent a) For $\alpha_1 \in \left[0 , \frac {L_c -  \alpha_2 Q_{2,\textrm{max}} }  {Q_1}\right]$, the total energy stored is below the critical load $L_c$ for all possible types of MGO $2$. Here, MGO $1$'s expected utility is given by
\vspace{-0.2cm}

 \begin{align*}\label{eq:3aa}
   E_{1,2a}(\boldsymbol{\alpha},Q_1) = \rho\left( Q_1 - \alpha_1 Q_1 \right) + \theta \rho_{c} \alpha_1 Q_1.
 \end{align*}
$E_{1,2a}$ is a strictly increasing function given that   $\theta \rho_{c} > \rho$, hence, it is maximized at its upper boundary
 $\alpha_{1,2a}^{*} = \frac {L_c -  \alpha_2 Q_{2,\textrm{max}} } {Q_1} $.\\
\vspace{-0.2cm}

 b) For $\alpha_1 \in \left[\frac {L_c -  \alpha_2 Q_{2,\textrm{max}} }  {Q_1} , 1\right] $, given MGO $2$'s  strategy,  the total energy stored is above the critical load for certain types of MGO $2$.  MGO $1$'s expected utility is given by
\vspace{-0.7cm}

\begin{multline}
\footnotesize
 E_{1,2b}(\boldsymbol{\alpha},Q_1) = \int_{0}^{A}   U_1(\boldsymbol{\alpha},\boldsymbol{Q}) f(Q_2) dQ_2 \\
\footnotesize +  \int_{A}^{Q_{2,\textrm{max}}}   U_1(\boldsymbol{\alpha},\boldsymbol{Q}) f(Q_2) dQ_2,
 \end{multline}
with $A =  \frac {L_c - \alpha_1 Q_1}  {\alpha_2 }$ which follows from (5).  Under this assumption,  $f_1(Q_2) = 1/Q_{2,\textrm{max}}$ over its domain and $E_{1,2b}$ is now given by

\vspace{-0.3cm}
\scriptsize
 \begin{multline} \label{eq:20} 
\scriptsize
 E_{1, 2b}(\boldsymbol{\alpha},Q_1) =  \frac {1}  {Q_{2,\textrm{max}}}  \int_{0}^{A}  \left[ \rho\left( Q_1 - \alpha_1 Q_1 \right) + \theta \rho_{c} \alpha_1 Q_1 \right] dQ_2 +\\
\scriptsize  \dfrac {1}  {Q_{2,\textrm{max}}}\int_{A}^{Q_{2,\textrm{max}}}  \left[ \rho \left( Q_1 - \alpha_1 Q_1\right) + \frac{1}{2} \theta \rho_{c} \left(\alpha_1 Q_1 -\alpha_2 Q_2 + L_c \right) \right]  dQ_2.
\end{multline}
\normalsize 

\vspace{-0.3cm}
By taking the second derivative of (\ref{eq:20}) with respect to the decision variable $\alpha_1$, we get

\vspace{-0.35cm}
\begin{align*} 
\frac {\partial E_{1,2b}} {\partial^2\alpha_1} = -  \frac{Q_1^2 \rho_c\theta} {2 \alpha_2Q_{\textrm{max},2} }.
\end{align*}
\vspace{-0.25cm}

\noindent The function is strictly concave given that its second derivative is strictly negative. The optimal solution is, hence, obtained by the necessary and sufficient optimality condition given by
\vspace{-0.2cm}

\begin{align} \label{eq:der} 
 \frac {\partial E_{1,2b}} {\partial\alpha_1} = 0.
\end{align}\vspace{0.02cm}

\noindent  (\ref{eq:der}) has a unique solution which is given by

 \begin{align}
\nonumber  \alpha_{1,r} =   \frac {L_c \rho_c \theta + ( \rho_c \theta - 2\rho)\alpha_2 Q_{2,\textrm{max}}  } { Q_1 \rho_c \theta}.
 \end{align}

\vspace{0.2cm}
Given that  $E_{1,2b}$ is a strictly concave function and that $\alpha_1$ is restricted to { $\left[\frac {L_c -  \alpha_2 Q_{2,\textrm{max}} }  {Q_1} \,, 1\right]$},
$\alpha_{1,2b}^{*}$ will be\\

\vspace{-0.25cm}
\begin{align}\label{eq:13b} \footnotesize 
    \alpha_{1,2b}^{*} =
    \begin{cases}
     \frac {L_c -  \alpha_2 Q_{2,\textrm{max}} }  {Q_1},  & \text{if}\   \alpha_{1,r} < \frac{L_c -  \alpha_2 Q_{2,\textrm{max}}}{Q_1}, \\
      \alpha_{1,r}, & \text{if}\  \alpha_{1,r} \in \left[\frac {L_c -  \alpha_2 Q_{2,\textrm{max}} }  {Q_1} \,\,\,\, 1\right], \\
      1, & \text{if}\  \alpha_{1,r} > 1. 
    \end{cases}
\end{align}

\vspace{0.15cm} 

 In fact, $\alpha_{1,r}$ is the optimal solution for $E_{1,2b}$ if it belongs to the feasible region of $E_{1,2b}$.  On the other hand, if $\alpha_{1,r}$ is larger than the upper bound, then  $E_{1,2b}$ is a strictly increasing function over the feasibility set and is maximized at its upper bound  $\alpha_{1,2b}^{*} = 1$. Finally, if $\alpha_{1,r}$ is smaller than the domain's lower bound $\frac {L_c -  \alpha_2 Q_{2,\textrm{max}}}{Q_1}$, then  $E_{1,2b}$ is a strictly decreasing function over the feasibility set and is maximized at its lower bound. However, the condition $\alpha_{1,r} < \frac{L_c -  \alpha_2 Q_{2,\textrm{max}}}{Q_1}$ cannot be satisfied for $\rho_c \theta > \rho$, and thus $\frac {L_c -  \alpha_2 Q_{2,\textrm{max}} }  {Q_1}$ cannot be the maximizer of $E_{1,2b}$. We can thus rewrite (\ref{eq:13b}) as

\begin{align}\label{eq:brbrbrbrbr} 
 \alpha_{1,2b}^{*} =
   \begin{cases}
    \alpha_{1,r}, & \text{if}\   \, \left[\frac {2\rho}{\rho_c\theta}-1\right] \alpha_2 > \frac{L_c -  Q_1} {Q_{2,\textrm{max}}},  \\
    1, & \text{if}  \, \left[\frac {2\rho}{\rho_c\theta}-1\right] \alpha_2 \leq \frac{L_c -  Q_1} {Q_{2,\textrm{max}}}.\\
    \end{cases}
\end{align}
\vspace{+0.1cm} 


We first note that $E_{1, 2a} = E_{1,2b}$ for  $\alpha_{1} = \frac {L_c -  \alpha_2 Q_{2,\textrm{max}}}{Q_1}$ which is the maximizer of $E_{1, 2a}$. However, as previously discussed,
$E_{1, 2b}$ cannot be maximized at $ \frac {L_c -  \alpha_2 Q_{2,\textrm{max}}}{Q_1}$. Thus,  the maximizer of MGO $1$'s expected utility, for $\alpha_2 \in \left[\frac {L_c -  Q_1} {Q_{2,\textrm{max}}}, 1 \right]$, belongs to the domain $\left[\frac {L_c -  \alpha_2 Q_{2,\textrm{max}}}{Q_1},1\right].$ In other words,
$r_{1}(\alpha_2) = \alpha_{1,2b}^{*}$  for $\alpha_2 \in \left[\frac {L_c -  Q_1} {Q_{2,\textrm{max}}}, 1 \right]$.

\section{Proof of Proposition 5}
\linespread{0.95}
 Player 1's expected utility under PT, for $\alpha_2 \in \left[\frac {L_c -  Q_1} {Q_{2,\textrm{max}}}, 1 \right]$, and
 $\alpha_1 \in \left[\frac {L_c -  \alpha_2 Q_{2,\textrm{max}} }  {Q_1} , 1\right] $, is given by
\vspace{-0.2cm}
\scriptsize
 \begin{multline}\label{eq:3b3b}
 E_{\textrm{PT},1,2b}(\boldsymbol{\alpha},Q_1) =   \int_{0}^{A} \dfrac {1}  {Q_{2,\textrm{max}}} V \left( \rho\left( Q_1 - \alpha_1 Q_1 \right) + \theta \rho_{c} \alpha_1 Q_1 \right) dQ_2 +\\
 \int_{A}^{Q_{2,\textrm{max}}}   \dfrac {1}  {Q_{2,\textrm{max}}}V\left( \rho Q_1 \left(1- \alpha_1\right) + \dfrac {1}{2} \theta \rho_{c} \left(\alpha_1 Q_1 -\alpha_2 Q_2 + L_c \right) \right)  dQ_2.
\end{multline}
\normalsize 

 We denote by $I_1$  the first integral in  ($\ref{eq:3b3b}$), and by $I_2$ the second. As previously mentioned, PT states that a utility is perceived in terms of gains and losses with respect to the reference point. Next, we analyze the possible values of both integrals $I_1$ (first integral) and $I_2$ (second integral)  in (\ref{eq:3b3b}) from that perspective.
The original utility in $I_1$, $U_{I,1}=\rho\left( Q_1 - \alpha_1 Q_1 \right) + \theta \rho_{c} \alpha_1 Q_1$, is only a function of $\alpha_1$ and is independent of  $Q_2$.  Equation (\ref{eq:CPI1}) follows from the fact that for {$\alpha_1 \leq B$},  $U_{I,1}$   is below the reference point $R_1$ and is thus perceived as a loss. On the other hand, it is considered as a gain for $\alpha_1 > B$.

We then assess the possible values of $I_2$. The original utility function in $I_2$, $U_{I,2} = \rho\left( Q_1 - \alpha_1 Q_1\right) + \frac{1}{2} \theta \rho_{c} \left(\alpha_1 Q_1 -\alpha_2 Q_2 + L_c \right) $ is considered a loss given that
\vspace{-0.2cm}

 \begin{equation}\label{eq:Q_r} \nonumber
\rho\left( Q_1 - \alpha_1 Q_1\right) + \dfrac {1}{2} \theta \rho_{c} \left(\alpha_1 Q_1 -\alpha_2 Q_2 + L_c \right) < R_1,
 \end{equation}
which can be rewritten as $Q_{2,r} < Q_2$ with $Q_{2,r}$ given by

\begin{equation}\label{eq:Q_r_2} \small
Q_{2,r} =   \dfrac{2}{\rho_c \theta \alpha_2} \left[ \rho\left( Q_1 - \alpha_1 Q_1\right) + \dfrac {1}{2} \theta \rho_{c} \left(\alpha_1 Q_1  + L_c \right) - R_1 \right]. 
\end{equation}
\normalsize

Given that MGO $1$'s expected utility is taken over MGO $2$'s type ($Q_2$), we next analyze $I_2$ for different values of $Q_2$. (\ref{eq:CPI2}) follows from the fact that $I_2$ is a loss integral for $Q_{2,r} < A$. Given that the lower bound of  $I_2$ is larger than $A$, then the entire range of $Q_2$ values is as well. The condition $Q_{2,r} < A$ can be rewritten as $C_1$. On the other hand, $I_2$ is a gain integral for $Q_{2,r} > Q_{2,\textrm{max}}$ which can be rewritten as $C_2$. Finally, for  $A < Q_{r,2} < Q_{2,\textrm{max}}$,   $I_2$ is split into two parts: a gain integral on  $\left[ A , Q_{2,r}  \right]$ and a loss integral on $\left[ Q_{\textrm{ref}2},  Q_{2,\textrm{max}}  \right]$.  $A < Q_{2,r} < Q_{2,\textrm{max}}$ can be rewritten as $C_3$. (\ref{eq:CPI1}) and (\ref{eq:CPI2}) are obtained by evaluating the integrals $I_1$ and $I_2$ for the described cases.
\end{appendices}


%


\end{document}